\documentclass[sigconf]{acmart}

\usepackage{booktabs} 

\usepackage[ruled]{algorithm2e} 

\SetAlFnt{\small}
\SetAlCapFnt{\small}
\SetAlCapNameFnt{\small}
\SetAlCapHSkip{0pt}
\IncMargin{-\parindent}

\newcommand{\remove}[1]{}
\newcommand{\h}{\hspace*{0.2in}}
\newcommand{\Ra}{\Rightarrow}
\newcommand{\ra}{\rightarrow}
\newcommand{\CC}{L}

\newcommand{\RR}{\mathbf{R}}

\DeclareMathOperator{\frontier}{frontier}
\DeclareMathOperator{\forbidden}{forbidden}

\usepackage{microtype}

\graphicspath{{./figs/}}

\acmConference[ICDCN'23]{ICDCN'2023}{January 2023}{Kharagpur,  India}
\acmYear{2023}
\copyrightyear{2023}

\begin{document}
\title{Lattice Linear Predicate Algorithms for the Constrained Stable Marriage Problem with Ties}
\titlenote{partially supported by NSF CNS-1812349, and the Cullen Trust for Higher Education Endowed Professorship}
\author{Vijay K. Garg}
\orcid{1234-5678-9012}
\affiliation{%
  \institution{The University of Texas at Austin}
  \city{Austin}
  \state{Texas}
 \postcode{78712}
    \country{USA}
}
\email{garg@ece.utexas.edu}


\bibliographystyle{ACM-Reference-Format}


\begin{abstract}
We apply Lattice-Linear Predicate Detection Technique to derive parallel and distributed algorithms
for various variants of the stable matching problem. These problems are:
(a) the constrained stable marriage problem (b) the super stable marriage problem in presence of ties, 
and (c) the strongly stable marriage in presence of ties. All these problems are solved using the 
Lattice-Linear Predicate (LLP) algorithm showing its generality.
The constrained stable marriage problem is a version of finding the stable marriage in presence of
lattice-linear constraints such as ``Peter's regret is less than that of Paul.''
For the constrained stable marriage problem, we present a distributed algorithm
that takes $O(n^2)$ messages each of size $O(\log n)$ where $n$ is the number of men in the problem.
Our algorithm is completely asynchronous. Our algorithms for the stable marriage problem with ties
are also parallel  with no synchronization. 

 \end{abstract}

\maketitle



\section{Introduction}
The Lattice-Linear Predicate (LLP) algorithm \cite{DBLP:conf/spaa/Garg20} is a general technique for designing parallel algorithms for combinatorial
optimization problems. In \cite{DBLP:conf/spaa/Garg20}, it is shown that the stable marriage problem, the shortest path problem in a graph, and 
the assignment problem can all be solved using the LLP algorithm.
In \cite{Garg:ICDCN22},  many dynamic programming problems, in \cite{DBLP:conf/sss/Garg21}, the housing problem, and
in \cite{AlvGar22}, it is shown that the minimum spanning tree problem can be solved using the LLP
algorithm.
In \cite{DBLP:conf/sss/GuptaK21}, Gupta and Kulkarni extend LLP algorithms for deriving self-stabilizing algorithms.
In this paper, we show that many generalizations of the stable matching problem can also be solved using
the LLP algorithm. A forthcoming book on parallel algorithms \cite{Garg23} gives a uniform description of these and other problems
that can be solved using the LLP algorithm.

The Stable Matching Problem (SMP) \cite{gale1962college} has wide applications in economics, distributed computing, resource allocation and many other fields \cite{maggs2015algorithmic,iwama2008survey}.
In the standard SMP, there are $n$ men and $n$ women each with their totally ordered preference list. 
The goal is to find 
a matching between men and women such that there is no instability, i.e., there is no pair of a woman and a man such that they are not married to each other but
prefer each other over their partners. 
In this paper, we show that LLP algorithm can be used to derive solutions to a more general problem than SMP, called {\em constrained SMP}. In our formulation, in addition to men's preferences and women's preferences, there may be a set of
{\em lattice-linear} 
constraints 
on the set of marriages consistent with men's preferences. 
For example, we may state that Peter's regret  \cite{gusfield1989stable} should be less than that of Paul,
where the {\em regret} of a man in a matching is the choice number he is assigned. As another example, we may require
the matching must contain some pairs called  {\em forced pairs}, or must not contain some pairs called {\em forbidden pairs} \cite{Dias2003}.
We call such constraints {\em external} constraints. Any algorithm to solve constrained SMP can solve standard SMP by setting (external) constraints to the empty set.

In this paper, 
we also present a distributed algorithm to solve the constrained SMP in an asynchronous system. One of the goals is to show how a parallel LLP algorithm
can be converted into a distributed asynchronous algorithm.
Our distributed algorithm uses a diffusing computation whose
termination is detected using a standard algorithm such as the Dijkstra-Scholten algorithm. The algorithm uses
$O(n^2)$ messages each of size $O(\log n)$.
Kipnis and Patt-Shamir \cite{KipnisP10} have given a distributed algorithm for stable matching in a synchronous system. There are many differences with their work. First, they do not
consider external constraints and their work is not easily extensible for incorporating external constraints.
Second, for termination detection, they require each rejected node to broadcast the fact that the protocol has not terminated on a shortest-path tree. This step
requires the assumption of synchrony for termination detection and incurs additional message overhead. Our algorithm avoids such broadcasts and works for asynchronous systems.
Their paper suggests use of $\alpha$ synchronizer \cite{JACM::Awerbuch1985} for simulating in asynchronous systems. However, each round adds $O(n^2)$ messages for using
$\alpha$ synchronizer. Thus, our algorithm not only solves a more general problem, it is also more efficient for running the traditional SMP in an asynchronous system.

We also consider the generalizations of the stable matching problem to the case when the preference lists may have ties.
The problem of stable marriage with ties is clearly more general than the standard stable matching problem and has also been
extensively studied \cite{IRVING1994261,gusfield1989stable,david2013algorithmics}. 
We consider three versions of matching with ties. 
In the first version, called {\em weakly stable} matching $M$, there is no blocking pair of 
man and woman $(m,w)$ who are not married in $M$ but strictly prefer
each other to their partners in $M$. In the second version, called {\em superstable} matching $M$, we require
that there is no blocking pair of man and woman $(m,w)$ who are not married in $M$
but either (1) both of them prefer each other to their partners in $M$, or (2) one of them prefers the other over his/her partner in $M$ and the other one is
indifferent, or (3) both of them are indifferent to their spouses.
The third version, called {\em strongly stable matching}, we require that
if there is no blocking pair $(m,w)$ such that they are not married in $M$
but either (1) both of them prefer each other to their partners in $M$, or (2) one of them prefers the other over his/her partner in $M$ and the other one is
indifferent. Algorithms for these problems are well-known; our goal is to 
present LLP algorithms for these problems.

\section{Background: Lattice-Linear Predicate Detection Algorithm}
In this section, we give a self-contained description of the Lattice-Linear Predicate detection algorithm.
The reader should consult \cite{DBLP:conf/spaa/Garg20} for more details.
Let $\CC$ be the lattice of all $n$-dimensional vectors of reals greater than or equal to zero vector and less than or equal to a given vector $T$
where the order on the vectors is defined by the component-wise natural $\leq$.
The lattice is used to model the search space of the combinatorial optimization problem.
The combinatorial optimization problem is modeled as finding the minimum element in $\CC$ that satisfies a boolean {\em predicate} $B$, where
$B$ models {\em feasible} (or acceptable solutions).
We are interested in parallel algorithms to solve the combinatorial optimization problem with $n$ processes.
We will assume that the systems maintains as its state the current candidate vector $G \in \CC$ in the search lattice, 
where $G[ i]$ is maintained at process $i$. We call $G$, the global state, and $G[ i]$, the state of process $i$.

Fig. \ref{fig:poset-lattice} shows a finite poset corresponding to $n$ processes ($n$ equals two in the figure), and the corresponding lattice of all eleven global states.

\begin{figure}[htbp]
\begin{center}
\includegraphics[width=2.5in,height=1.0in]{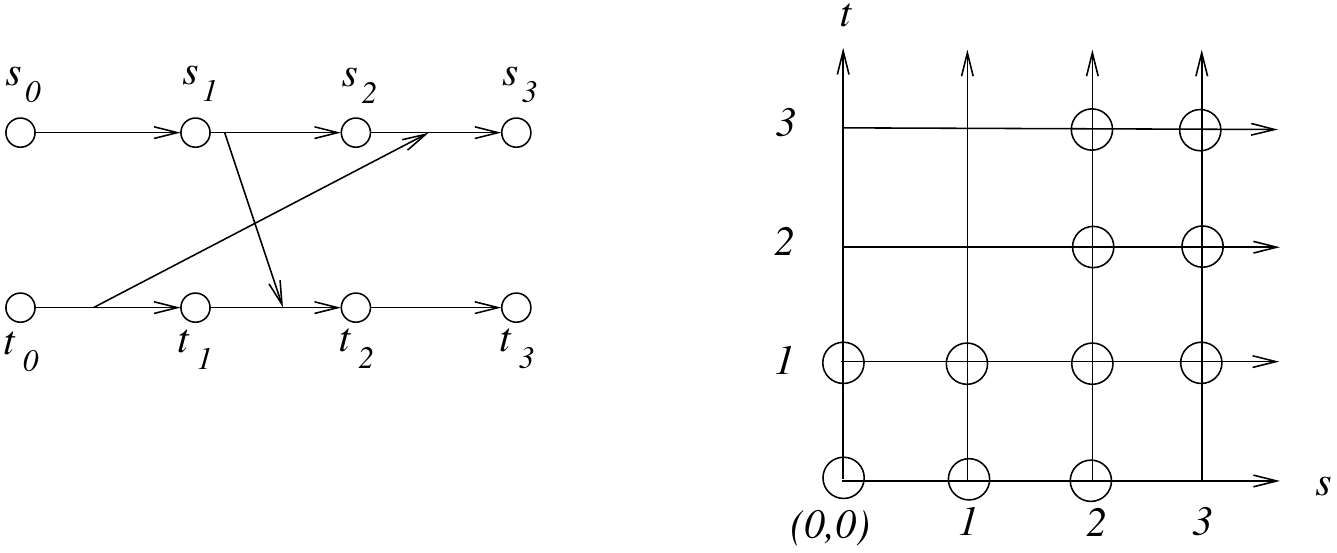}
\caption{A poset and its corresponding distributive lattice $L$ \label{fig:poset-lattice}}
\end{center}
\end{figure}

Finding an element in lattice that satisfies the given predicate $B$, is called the {\em predicate detection} problem.
Finding the {\em minimum} element that satisfies $B$ (whenever it exists) is the combinatorial optimization problem.
A key concept in deriving an efficient  predicate detection algorithm is that of a {\em
forbidden} state.  
Given a predicate $B$, and a vector $G \in \CC$, a state $G[j]$ is {\em forbidden} (or equivalently, the index $j$ is forbidden) if 
for any vector $H \in \CC$ , where $G \leq H$, if $H[j]$ equals $G[ j]$, then $B$ is false for $H$.
Formally,
\begin{definition}[Forbidden State  \cite{chase1998detection}]
  Given any distributive lattice $\CC$ of  $n$-dimensional vectors of $\RR_{\ge 0}$, and a predicate $B$, we define
  $ \forbidden(G,j,B) \equiv \forall H \in \CC : G \leq H : (G[j] = H[j]) \Rightarrow
  \neg B(H).$
\end{definition}

We define a predicate $B$ to
be {\em lattice-linear} with respect to a lattice $\CC$
 if for any global state $G$,  $B$ is false in $G$ implies that $G$ contains a
{\em forbidden state}. Formally,
\begin{definition}[lattice-linear Predicate  \cite{chase1998detection}]
A boolean predicate $B$ is {\em {lattice-linear}} with respect to a lattice $\CC$
iff
$\forall G \in \CC: \neg B(G) \Ra (\exists j: \forbidden(G,j,B))$.
\end{definition}

Once we determine $j$ such that $forbidden(G,j,B)$, 
we also need to determine how to advance along index $j$.
To that end, we extend the definition of forbidden as follows.
\begin{definition}[$\alpha$-forbidden]
 Let $B$ be any boolean predicate on the lattice $\CC$ of all assignment vectors.
 For any $G$, $j$ and positive real $\alpha > G[ j]$, we define $\mbox{forbidden}(G,j, B, \alpha)$ iff
 $$  \forall H \in \CC:H \geq G: (H[j] < \alpha) \Ra \neg B(H).  $$
\end{definition}

Given any lattice-linear predicate $B$, suppose $\neg B(G)$. This means that $G$ must be advanced on all
indices $j$ such that $\forbidden(G,j,B)$.  We use a function $\alpha(G,j, B)$ such that $\forbidden(G, j, B, \alpha(G,j, B))$ holds
whenever $\forbidden(G,j,B)$ is true.  With the notion of $\alpha(G, j, B)$, we have the Algorithm $LLP$.
The algorithm $LLP$ has two inputs --- the predicate $B$ and the top element of the lattice $T$. It returns the least vector $G$ which is less than or equal to $T$
and satisfies $B$ (if it exists). Whenever $B$ is not true in the current vector $G$, the algorithm advances on all forbidden indices $j$
in parallel. This simple parallel algorithm can be used to solve a large variety of combinatorial optimization problems
by instantiating different $\forbidden(G,j,B)$ and $\alpha(G,j,B)$.

\begin{algorithm}
\SetAlgoRefName{LLP}
 vector {\bf function} getLeastFeasible($T$: vector, $B$: predicate)\\
 {\bf var} $G$: vector of reals initially $\forall i: G[ i] = 0$;\\
 {\bf while} $\exists j: \forbidden(G,j,B)$ {\bf do}\\
    \h    {\bf for all} $j$ such that $\forbidden(G,j,B)$  {\bf in parallel}:\\
    \h\h {\bf if} $(\alpha(G,j,B) > T[j])$ then return null; \\
   \h\h    {\bf else} $G[ j] := \alpha(G,j,B)$;\\
       {\bf endwhile};\\
     {\bf return} $G$; // the optimal solution
\caption{Find the minimum vector at most $T$ that satisfies $B$\label{fig:alg-llp}}
\end{algorithm}

The following Lemma is useful in proving lattice-linearity of predicates.
\begin{lemma}\label{lem:basic-LLP}  \cite{DBLP:conf/spaa/Garg20,chase1998detection}
Let $B$ be any boolean predicate defined on a lattice $\CC$ of vectors. \\
(a) Let $f:\CC \ra \RR_{\ge 0}$ be any monotone function  defined on the lattice $\CC$ of vectors of $\RR_{\ge 0}$.
Consider the predicate
$B \equiv G[ i] \geq f(G)$ for some fixed $i$. Then, $B$ is lattice-linear.\\
(b) If $B_1$ and $B_2$ are lattice-linear then $B_1 \wedge B_2$ is also lattice-linear.
\end{lemma}

We now give an example of lattice-linear predicates for scheduling of $n$ jobs. Each job $j$ requires time $t_j$ for completion and has a set of
 prerequisite jobs, denoted by $pre(j)$, such that it can be started only after all its prerequisite jobs
have been completed. Our goal is to find the minimum completion time for each job.
We let our lattice $\CC$ be the set of all possible completion times. A completion vector $G \in \CC$ is feasible iff $B_{jobs}(G)$ holds where
$B_{jobs}(G) \equiv \forall j: (G[ j] \geq t_j) \wedge (\forall i \in pre(j): G[ j] \geq G[ i] + t_j)$.
$B_{jobs}$ is lattice-linear because if it is false, then there exists $j$ such that 
either $G[ j] < t_j$ or $\exists i \in pre(j): G[ j] < G[ i]+t_j$. We claim that $\forbidden(G, j, B_{jobs})$. Indeed, any vector $H \geq G$ cannot
be feasible with $G[ j]$ equal to $H[j]$. The minimum of all vectors that satisfy feasibility corresponds to the minimum completion time.

As an example of a predicate that is not lattice-linear, consider the predicate $B \equiv \sum_j G[ j] \geq 1$ defined on the space of 
two dimensional vectors. Consider the vector $G$ equal to $(0,0)$. The vector $G$ does not satisfy $B$. For $B$ to be lattice-linear
either the first index or the second index should be forbidden. However, 
none of the indices are
forbidden in $(0,0)$. The index $0$ is not
forbidden because the vector $H = (0,1)$ is greater than $G$, has $H[0]$ equal to $G[ 0]$ but it still satisfies $B$.
The index $1$ is also not forbidden
because $H =(1,0)$ is greater than $G$, has $H[1]$ equal to $G[ 1]$ but it satisfies $B$.

 \label{sec:prog}
 
 We now go over the notation used in description of our parallel algorithms.
Fig. \ref{fig:examples} shows a parallel algorithm for the job-scheduling problems.

The {\bf var} section gives the variables of the problem.
We have a single variable $G$ in the example shown in Fig. \ref{fig:examples}. 
 $G$ is an array of objects such that
 $G[j]$ is the state of thread $j$ for a parallel program.

The {\bf input} section gives all the inputs to the problem. These inputs are constant in the program and do not change during execution.

 The {\bf init} section is used to initialize the state of the program.
 All the parts of the program
 are applicable to all values of $j$. For example, the {\em init} section of the job scheduling program in Fig. \ref{fig:examples}
 specifies that $G[j]$ is initially $t[j]$. Every thread $j$ would initialize $G[j]$.

 The {\bf always} section defines additional variables which are derived from $G$.
 The actual implementation of these variables are left to the system. They can be viewed as
 macros. We will show its use later.

 The LLP algorithm gives the desirable predicate either by using the {\bf forbidden} predicate or {\bf ensure} predicate.
 The {\em forbidden} predicate has an associated {\em advance} clause that specifies how $G[j]$ must be advanced
 whenever the forbidden predicate is true.
 For many problems, it is more convenient to use the complement of the forbidden predicate.
 The {\em ensure} section specifies the desirable predicates of the form $(G[j] \geq expr)$ or 
 $(G[j] \leq expr)$. 
 The statement {\em ensure} $G[j] \geq expr$  simply means that whenever thread $j$ finds $G[j]$ to be less than 
 $expr$; it can advance $G[j]$ to $expr$.
 %
  Since $expr$ may refer to $G$, just by setting $G[j]$ equal to $expr$, there is no guarantee 
  that $G[j]$ continues to be  equal to $expr$ --- the value of $expr$ may change because of changes in other components.
  We use {\em ensure} statement whenever $expr$ is a monotonic function of $G$ and therefore the predicate
  is lattice-linear. 

\begin{figure}
\begin{center}
\small {
\fbox{\begin{minipage}  {\textwidth}\sf
\begin{tabbing}
xxxx\=xxxx\=xxxx\=xxxx\=xxxx\=xxxx\= \kill
$P_j$: Code for thread $j$\\
{\bf var} $G$: array[$1$..$n$] of $0..maxint$;// shared among all threads\\
{\bf input}: $t[j]: int$, $pre(j)$: list of $1..n$;\\
{\bf init}: $G[j] := t[j]$;\\
\\
{\bf \underline{job-scheduling}}:\\
\> {\bf forbidden}: $G[j] < \max \{G[i] + t[j] ~|~ i \in pre(j)\}$;\\
\> \> {\bf advance}: $G[j] := \max \{G[i] + t[j] ~|~ i \in pre(j)\}$;\\
\\
{\bf \underline{job-scheduling}}:\\
\> {\bf ensure}: $G[j] \geq \max \{G[i] + t[j] ~|~ i \in pre(j)\}$;\\
\\
{\bf \underline{shortest path from node $s$: Parallel Bellman-Ford}} \\
\> {\bf input}: $pre(j)$: list of $1..n$; $w[i,j]$: int for all $i \in pre(j)$\\
\> {\bf init}: if $(j=s)$ then $G[j] := 0$ else $G[j]$ := maxint;\\
\> {\bf ensure}:  $ G[j] \leq \min \{ G[i] + w[i,j] ~|~ i \in pre(j) \}$
\end{tabbing}
\end{minipage}
 }
 }
\caption{LLP Parallel Program for (a) job scheduling problem using forbidden predicate (b) job scheduling problem using ensure clause and (c) the shortest path problem
\label{fig:examples}}
\end{center}
\vspace*{-0.2in}
\end{figure}

\section{A Parallel Algorithm for the Constrained Stable Matching Problem}
 We now derive the algorithm for the stable matching problem using Lattice-Linear Predicates \cite{DBLP:conf/wdag/Garg17}.
We let $G[i]$ be the choice number that man $i$ has proposed to. Initially, $G[i]$ is $1$ for all men. 

 \begin{definition}
 An assignment $G$ is feasible for the stable marriage problem if (1) it corresponds to a perfect matching (all men are paired with different women) 
 and (2) it has no blocking pairs. 
 \end{definition}
 
 The predicate ``$G$ is a stable marriage'' is a lattice-linear predicate \cite{DBLP:conf/spaa/Garg20}  which
immediately gives us \ref{fig:GS}.
 The {\bf always} section defines variables which are derived from $G$.
 These variables can be viewed as
macros. For example, for any thread  $z = mpref[j][G[j]]$. This means that
whenever $G[j]$ changes, so does $z$.
If man $j$ is forbidden, it is clear that any vector in which man $j$ is matched with $z$ and the other man $i$ is matched
 with his current or a worse choice can never be a stable marriage. Thus, it is safe for man $j$ to advance to the next choice.
 \begin{algorithm}
 \SetAlgoRefName{LLP-ManOptimalStableMarriage}
$P_j$: Code for thread $j$\\
 {\bf input}: $mpref[i,k]$: int for all $i,k$; $wrank[k][i]$: int for all $k,i$;\\
  {\bf init}:   $G[j] := 1$;\\
{\bf always}:  $z = mpref[j][G[j]];$\\
 {\bf forbidden}: \\
  $\exists i:\exists k \leq G[i]:(z = mpref[i][k]) \wedge (wrank[z][i] < wrank[z][j])$\\
 \h  {\bf advance}: $G[j] := G[j]+1; $
\caption{A Parallel Algorithm for Stable Matching \label{fig:GS}}
\end{algorithm}

We now generalize \ref{fig:GS} algorithm to solve the constrained stable marriage problem.
In the standard stable matching problem, there are no constraints
on the order of proposals made by different men. 
Let $E$ be the set
of proposals made by men to women.
We also call these proposals {\em events}
which are executed by
$n$ processes corresponding to $n$ men denoted by $\{P_1 \ldots P_n\}$.
Each of the events can be characterized by a tuple $(i,j)$ that corresponds to
the proposal made by man $i$ to woman $j$.
We impose a partial order $\ra_p$ on this set of events to model the order in which
these proposals can be made. In the standard SMP, every man $P_i$ has its preference list $mpref[i]$ such that
$mpref[i][k]$ gives the $k^{th}$ most preferred woman for $P_i$. We model $mpref$ using
$\ra_p$; if $P_i$ prefers woman $j$ to woman $k$, then
there is an edge from the event $(i,j)$ to the event $(i,k)$. As in SMP, we assume that
every man gives a total order on all women.
Each process makes proposals to women in the decreasing order of preferences (similar to Gale-Shapley algorithm).

In the standard stable matching problem, there are no constraints
on the order of proposals made by different men, and $\ra_p$ can be visualized as a partial order $(E, \ra_p)$ with
$n$ disjoint chains.
We generalize the SMP problem to include external constraints
on the set of proposals.
In the constrained SMP, $\ra_p$  can relate proposals made by different men
and therefore $\ra_p$ forms a general poset $(E, \ra_p)$.
For example, the constraint that Peter's regret is less than or equal to John can be modeled by adding $\ra_p$ edges
as follows. For any regret $r$, we add an $\ra_p$ edge from the proposal by John with regret $r$ to the
proposal by Peter with regret $r$.
We draw $\ra_p$ edges in solid 
edges as shown in Fig. \ref{fig:csmp-model}.




Let $G \subseteq E$ denote the  global state of the system. A global state $G$ is simply the subset of events executed in the computation
such that it preserves the order of events within each $P_i$. 
Since all events executed by a process $P_i$ are totally ordered,
 it is sufficient to record the number of events executed 
by each process in a global state. 
Let $G[i]$ be the number of proposal made by $P_i$. 
Initially, $G[i]$ is $1$ for all men.
If $P_i$ has made $G[i]>0$ proposals, then 
$mpref[i][G[i]]$ gives the identity of the woman last proposed by $P_i$. 
We let $event(i, G[i])$ denote the event in which $P_i$ makes a proposal to $mpref[i][G[i]]$.
We also use $succ(event(i, G[i]))$ to denote the next proposal made by $P_i$, if any.


For the constrained SMP, we have $\ra_p$ edges that relate proposals of different processes.
The graph in Fig. \ref{fig:csmp-model} shows an example of using $\ra_p$ edges in the constrained SMP.
For this problem, we work with {\em consistent global states} (or order ideals \cite{davey,Gar:2015:bk}).
A global state $G \subseteq E$ is {\em consistent} if
$ \forall e,f \in E: (e \ra_p f) \wedge (f \in G) \Ra (e \in G).$
In the context of constrained SMP, it is easy to verify that $G$ is consistent iff 
for all $j$, there does not exist $i$ such that $$succ(event(j, G[j])) \ra_p event(i, G[i]).$$
It is well known that the set of all consistent global states of a finite poset forms a finite
distributive lattice \cite{davey,Gar:2015:bk}. We use the lattice of all consistent global states as $\CC$ for
the predicate detection.

In the standard SMP, women's preferences
are specified by preference lists $wpref$ such that $wpref[i][k]$ gives the $k^{th}$ most preferred man for woman $i$.
It is also convenient to define $wrank$ such that $wrank[i][j]$ gives the choice number $k$ for which $wpref[i][k]$ equals $j$, i.e.,
$wpref[i][k] = j$ iff $wrank[i][j] = k$.
We model these preferences using edges on the computation graph as follows. If an event $e$ 
corresponds to a proposal by $P_i$ to woman $q$ and she prefers $P_j$, then we add a dashed
 edge
from $e$ to the event $f$ that corresponds to $P_j$ proposing to woman $q$.
The set $E$ along with the dashed edges also forms a partial order $(E, \ra_w)$ where 
$e \ra_w f$ iff both proposals are to the same woman and that woman prefers the proposal
$f$ to $e$. 
With $((E, \ra_p), \ra_w)$ we can model any SMP specified using $mpref$ and $wpref$.

Figure \ref{fig:cuts} gives an example of a standard SMP problem in Fig. \ref{fig:SMP} in our model. To avoid cluttering the
figure, we have shown preferences of all men but preferences of only two of the women.
Fig \ref{fig:csmp-model} gives an example of a constrained SMP. Since both $\ra_p$ and $\ra_w$ are
transitive relations, we draw only the transitively reduced diagrams.

\begin{figure}
\begin{tabular}{l | l l l l |   l l   l | l l l l |}
mpref & & & &  & & & wpref \\
P1   &   w4 & w1 & w2 & w3  &  & & w1 & P4 & P1 & P3 & P2\\
P2   &   w2 & w3 & w1 & w4 &  & & w2 & P1 & P4 & P2 & P3\\
P3   &   w3 & w1 & w4 & w2 &  & & w3 & P1 & P2 & P4 & P3\\
P4   &   w2 & w4 & w3 & w1 &  & & w4 & P3 & P1 & P4 & P2\\
\end{tabular}
\caption{\label{fig:SMP}  Stable Matching Problem specified using men preference list (mpref) and
women preference list (wpref).}
\end{figure}

\begin{figure}[htbp]
\vspace*{-0.3in}
\begin{center}
\includegraphics[height=3.5in]{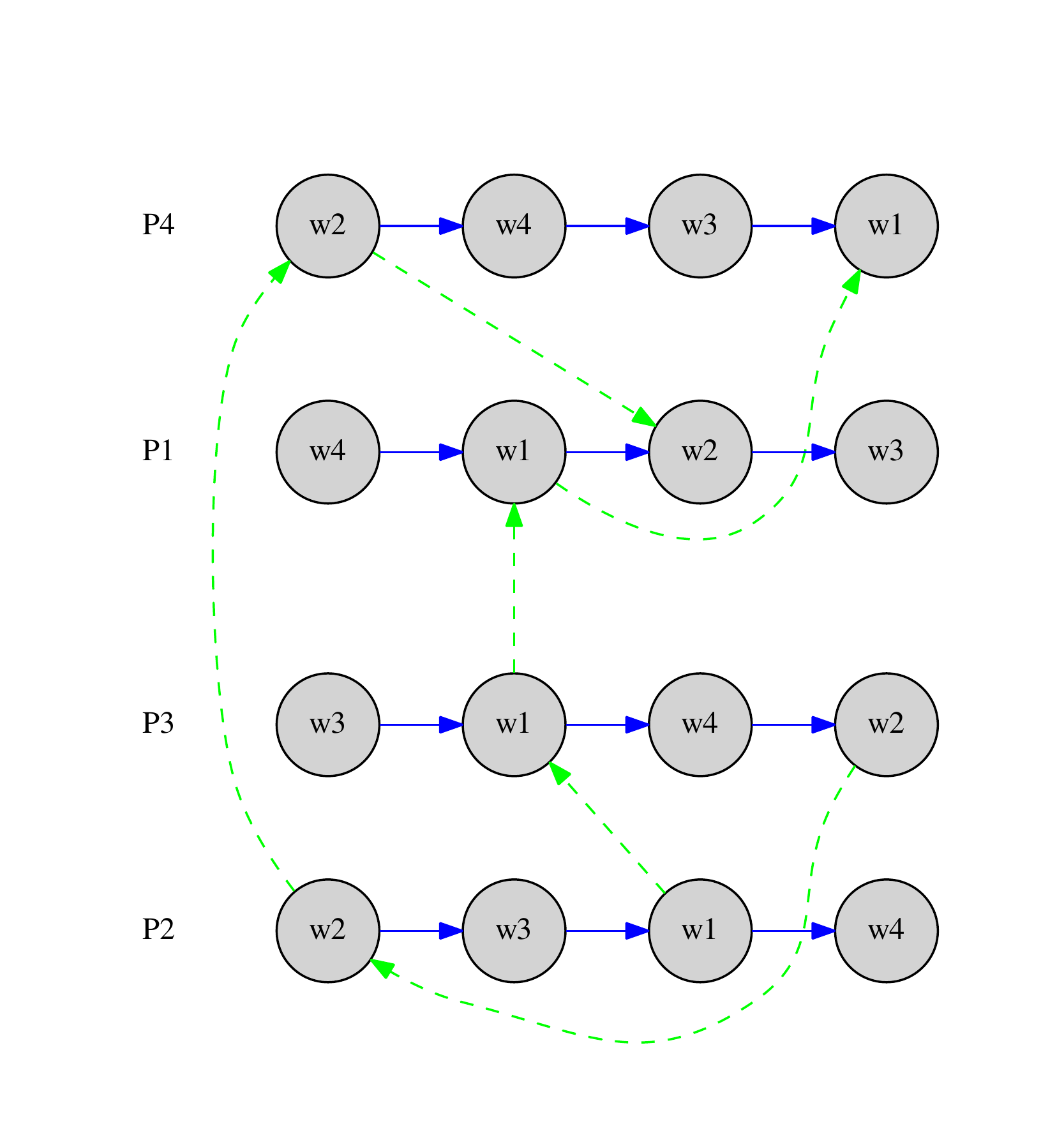}
\caption{\label{fig:cuts}  Men preferences are shown in blue solid edges. Preferences of women 1 and 2  are shown in dashed green edges.
In the standard SMP graph, there are no blue edges from any event in $P_i$ to any event in $P_j$ for distinct $i$ and $j$.}
\end{center}
\vspace*{-0.2in}
\end{figure}

\begin{figure}[htbp]
\vspace*{-0.3in}
\begin{center}
\includegraphics[height=3.5in]{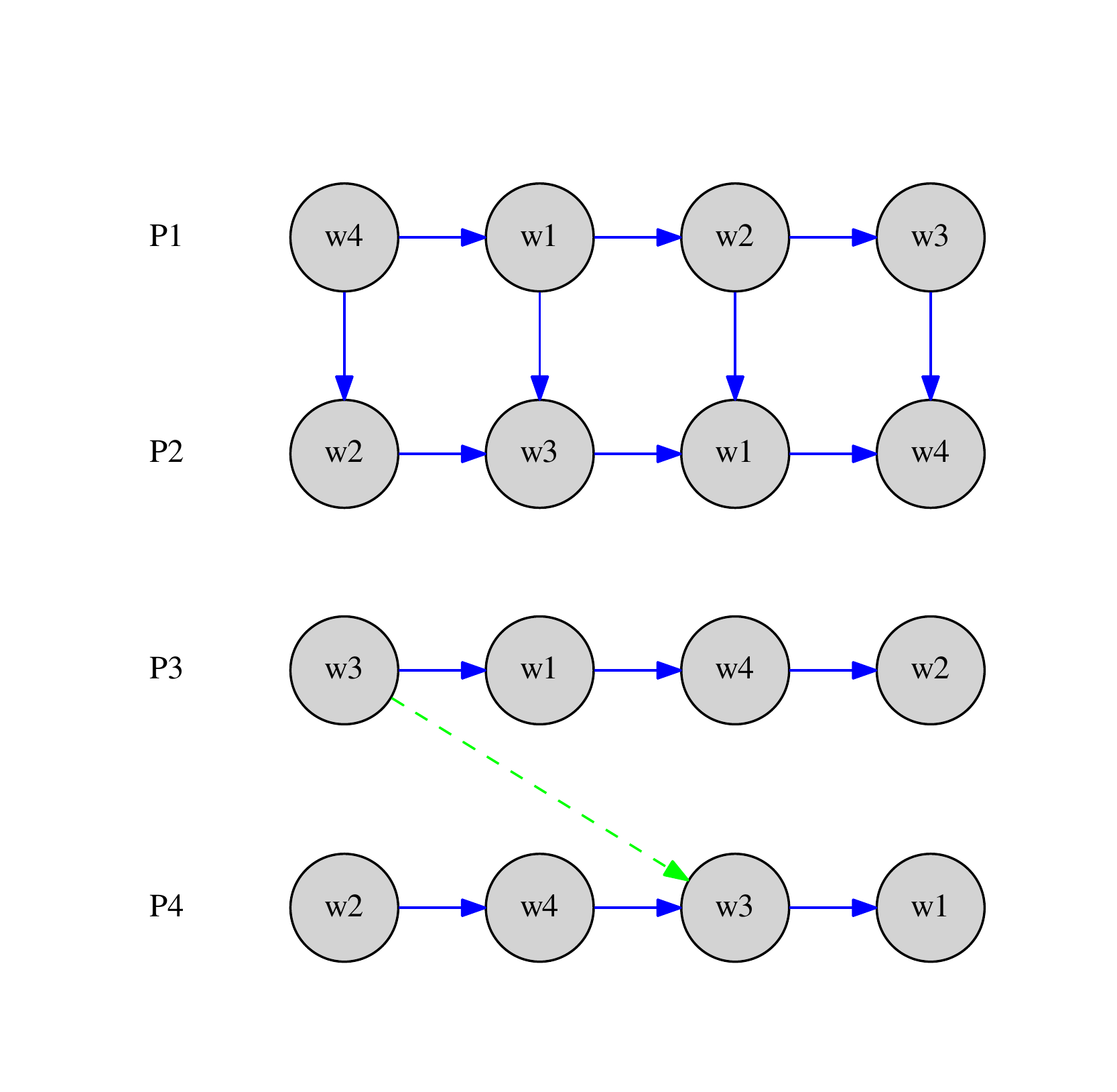}
\caption{\label{fig:csmp-model}  Constrained SMP Graph corresponding to constraint that the {\em regret} for $P_2$ is less than or equal to that of $P_1$.
It also shows the preference of $w3$ of $P4$ over $P3$.}
\end{center}
\vspace*{-0.2in}
\end{figure}



The above discussion motivates the following definition.
\begin{definition}[Constrained SMP Graph]
Let $E = \{(i,j) | i \in [1..n] \mbox{ and } j \in [1..n] \}$.
A Constrained SMP Graph $((E, \ra_p), \ra_w)$ is a directed graph on $E$ with two sets of edges $\ra_p$ and $\ra_w$ with the following
properties: (1) $(E, \ra_p)$ is a poset such that the set $P_i = \{ (i,j) | j \in [1..n] \}$ is a chain for all $i$, and
(2) $(E, \ra_w)$ is a poset such that the set $Q_j = \{ (i,j) | i \in [1..n] \}$ is a chain for all $j$  and there is no $\ra_w$ edge
between proposals to different women, i.e.,
for all $i,j,k,l: (i,j) \ra_w (k,l) \Ra (j=l)$.
\end{definition}

Given a global state $G$, we define the {\em frontier} of $G$ as the set of maximal events executed by 
any process. The frontier includes only the last event executed by $P_i$ (if any). Formally,
$\frontier(G) = \{ e \in G ~|~ \forall f \in G$ such that  $f \neq e$, $f$ and $e$ are executed by $P_i$ implies $f \ra_p e$ \}.
We call the events in $G$ that are not in $\frontier(G)$ as pre-frontier events.

We now define the feasible predicate on  global states as follows.
\begin{definition}[feasibility for marriage]
A global state $G$ is  feasible for marriage iff (1)  $G$ is a consistent global state, and (2) there is no dashed edge ($\ra_w$) from a frontier event to any event of $G$ (frontier or pre-frontier).
Formally,
$B_{marriage}(G) \equiv$\\
$ consistent(G) \wedge  (\forall e \in \frontier(G), \forall g \in G: \neg (e \ra_w g).$

\end{definition}

It is easy to verify that the problem of finding a stable matching is the same as finding a global state that 
satisfies the predicate $B_{marriage}$ which is defined purely in graph-theoretic terms on the
constrained SMP graph.
The next task is to show that $B_{marriage}$ is lattice-linear.

\begin{theorem}\label{lem:CSMP-forbid}
For any global state $G$ that is not a constrained stable matching,
there exists  $i$ such that $\forbidden(G,i,B_{marriage})$.
\end{theorem}
\begin{proof}

First suppose that $G$ is not consistent, i.e., there exists $f \in G$ such that
there exists $e \not \in G$ and $e \ra_p f$. Suppose that $e$ is on $P_i$. Then, 
$\forbidden(G,i,B)$ holds because any global state $H$ that is greater than $G$ cannot be consistent
unless $e$ is included.

Next, suppose that $G$ is a consistent global state but the assignment for $G$ is not a matching.
This means that for some distinct $i$ and $j$, both $G[i]$ and $G[j]$ refer to the same woman, say $w$.
Suppose that $w$ prefers $j$ to $i$, then we claim $\forbidden(G, i, B)$.
Consider any $H$ such that $H[i] = G[i]$ and $H[j] \geq G[j]$.
First consider the case $H[j] = G[j]$. In this case, the same woman $w$ is still assigned to 
two men and hence $H$ is not a stable matching. 
Now consider the case $H[j] > G[j]$. In this case, 
the woman $w$ prefers man $j$ to $i$,
and the man $j$ prefers $w$ to the woman assigned in $H[j]$ violating stability.

Now suppose that the assignment for $G$ is a constrained matching but not stable. 
Suppose that $(j,w)$ is a blocking pair in $G$. Let $i$ be assigned to $w$ in $G$
 (i.e., the woman corresponding to $G[i]$ prefers man $j$ to $i$,
and the man $j$ also prefers her to his assignment).
We claim that $\forbidden(G, i, B)$.
Consider any $H$ such that $H[i] = G[i]$ and $H[j] \geq G[j]$.
In this case, $(j,w)$ continues to be blocking in $H$.
The woman $w$ prefers man $j$ to $i$,
and the man $j$ prefers $w$ to the woman assigned in $H[j]$.

%
\end{proof}

We now apply the detection of lattice-linear global predicates for the constrained stable matching. 

 \begin{algorithm}
 \SetAlgoRefName{LLP-ConstrainedStableMarriage}
$P_j$: Code for thread $j$\\
 {\bf input}: $mpref[i,k]$: int for all $i,k$; $wrank[k][i]$: int for all $k,i$;\\
  {\bf init}:   $G[j] := 1$;\\
{\bf always}:  $z = mpref[j][G[j]];$\\
 {\bf forbidden}: \\
  $\exists i: \exists k \leq G[i]: (z = mpref[i][k]) \wedge (wrank[z][i] < wrank[z][j])$ $\vee (\exists i: succ(event(j, G[j])) \ra_p event(i, G[i]]))$ \\
\h  {\bf advance}: {\bf if} $(G[j] < n)$ then $G[j] := G[j]+1; $\\
\h \h {\bf else} print(``no constrained stable marriage'')
\caption{A Parallel Algorithm for the Constrained Stable Matching \label{fig:alg}}
\end{algorithm}

The algorithm to find the man-optimal constrained stable marriage is shown in Fig. \ref{fig:alg}.
From the proof of Theorem \ref{lem:CSMP-forbid}, we get the
following implementation of $\forbidden(G, j, B_{marriage})$ in Fig. \ref{fig:alg}.
The first disjunct holds when the woman $z$ assigned to man $j$ is such that there exists a man $i$
who is either (1) currently assigned to $z$ and woman $z$ prefers man $i$ , or (2) currently assigned
to another woman but he prefers $z$ to the current assignment. The first case holds when $k=G[i]$ and the
second case holds when $k < G[i]$.
The first case is equivalent to checking if a dashed edge exists from $(j, z)$ to a frontier event.
The second case is equivalent to checking if a dashed edge exists to a pre-frontier event.
The second disjunct checks that the assignment for $G$ satisfies all external constraints with respect to $j$.

Our algorithm generalizes the Gale-Shapley algorithm in that it allows specification of external constraints.

We now show an execution of the algorithm on the CSMP in Fig. \ref{fig:csmp-model}.
Since every $P_i$ must make at least one proposal, we start with 
the first proposal for every $P_i$. The corresponding assignment is
$[w_4, w_2, w_3, w_2]$, i.e., $P_1$ is assigned $w_4$, $P_2$ is assigned $w_2$ and so on.
In this global state $G$, the second component is forbidden. This is because
$w_2$ prefers $P_4$ over $P_2$. We advance on  $P_2$ to get the global state
$[w_4, w_3, w_3, w_2]$. Now, 
because $w_3$ prefers $P_2$ over $P_3$, $P_3$ must advance.
We get the global state
$[w_4, w_3, w_1, w_2]$. This is a stable matching. However, it does not satisfy
 the constraint that the regret of $P_2$ is less than or equal to that of  $P_1$.
 Here, $P_1$ is forbidden and $P_1$ must advance. We now get the global state
$[w_1, w_3, w_1, w_2]$ which is not a matching. 
Since $w_1$ prefers $P_1$ over  $P_3$, $P_3$ must advance.
We reach the global state
$[w_1, w_3, w_4, w_2]$ which satisfies the constrained stable matching.

%



We have discussed man-oriented constrained stable marriage problem. 
One can also get an LLP algorithm the for woman-oriented constrained stable marriage problem.
The paper \cite{GarHu20} gives an algorithm $\beta$ that does the downward traversal in the proposal lattice in search of a stable marriage.
When men and women are equal then such a traversal
can be accomplished by switching the roles of men and women. However, in \cite{GarHu20} is is assumed that
the number of men $n_m$ may be much smaller than the number of women $n_w$.  It has the
time complexity of $O(n_m^2 + n_w)$. Switching the roles of men and women is not feasible
without increasing the complexity of the algorithm.

\section{A Distributed Algorithm for the Constrained Stable Matching Problem}
Although the standard SMP has been studied in a distributed system setting (e.g., \cite{brito2005distributed,kipnis2009note}), 
we study the constrained SMP in a distributed system setting. 
Our goal is to show how a parallel LLP algorithm can be converted to a distributed program.
We assume an asynchronous system in which all channels are FIFO and reliable and
that processes do not crash.

We assume that each man and woman knows only his or her preference lists. 
$P_i$ corresponds to the computation at man  $i$ and $Q_i$ corresponds to the computation at woman $i$.
Each process $P_i$ is responsible for updating its own component in $G[i]$.
For the LLP algorithm, we will assume that the only variable at $P_i$ is $G$ and
all other variables such as $mpref$ are constants.
In addition, each man is given
a list of prerequisite proposals for each of the women that he can propose to.
In terms of the constrained-SMP graph, this corresponds to every man knowing the
incoming solid
edges for the chain that corresponds to that man in the graph.
From $mpref$, one can also derive $mrank$, the rank $P_i$ assigns to each woman.

The process $Q_i$ has $wpref$, preferences of woman $i$. However, it is more convenient to
keep $wrank$, the rank $Q_i$ assigns to each man. This information is input to $Q_i$.
The only variable a woman $Q_i$
maintains is the {\em partner}. Note that given $G$, the partner for each woman can be derived.
However, in a distributed system setting it is more efficient to maintain the partner at each woman.

Whenever $G[i]$ is updated by $P_i$, we will assume that $P_i$ sends a message
to other relevant processes informing them about the update. Each process keeps
enough information to be able to evaluate its forbidden predicate. Since the message
transfer takes time, the data structures are not necessarily up to date at each process.
In particular $P_j$ may have an old value of $G[i]$ maintained at $P_i$.
We show that the LLP algorithm has the advantage that it works correctly
despite the fact that processes use old values of $G$. Each process evaluates
its forbidden predicate and advances its state whenever the forbidden predicate is true.
The algorithm terminates when no process is forbidden.  In a distributed system setting, 
we need some process to determine that the system has reached such a state.
A possible solution for running LLP algorithms in a distributed environment is to 
run it as a diffusing computation\cite{dijkstra1980termination}  and use a termination detection algorithm
along with the LLP algorithm.

We now present a diffusing computation for
solving the constrained SMP.
We adopt the standard rules of a diffusing computation.
A {\em passive} process can become {\em active} only on receiving messages, and only an active process can send a message.
We assume the existence of a process called environment that starts the algorithm by
sending {\em initiate} messages to all men. In our algorithm shown in Fig. \ref{fig:CSMP},

There are four types of messages used in this algorithm. There are exactly $n$ {\em initiate} messages sent
by the environment to all men. Each man can send two types of messages. He sends {\em propose} messages
to women one at a time in the order given by $mpref$. These messages are sent whenever the current state of the man
is forbidden and he needs to advance to the next woman.
A man may sometimes skip proposing some women as
explained later. A man also sends {\em advance} messages to other men which may force other men to skip
certain proposals to satisfy external constraints. 

A woman acts only when she receives a {\em propose} message 
from a man $j$.
On receiving a {\em propose} message, if she is currently not engaged, she 
gets engaged to man $j$. If she is engaged to a man and the new proposal is preferable to her current
partner then she sends a {\em reject} message to the current partner. If the new proposal is less preferable, then
she sends a {\em reject} message to the proposer. The variable {\em partner} indicates her partner at any point.
If the value of {\em partner} is zero, then that woman is free; otherwise, she is engaged.
Note that a woman never sends any {\em accept} message.
The algorithm is based on the assumption that if a woman has received a proposal and not rejected it, then
she has accepted the message (the algorithm assumes that no messages are lost).

We now explain the behavior of men for each message type he receives as shown in Fig. \ref{fig:CSMP}.
On receiving an {\em initiate} message from the environment,  we know that any assignment must have
at least one proposal from that man. To satisfy external constraints, all proposals that are prerequisite must also be made.
Hence, the man sends an {\em advance} message to all men with prerequisite proposals. He then sends a proposal to his top choice.
On receiving a {\em reject} message, he first checks if the {\em reject} message is from his current partner. Since a man
may have advanced to a different proposal, there is no need for any action if the {\em reject} message is from an earlier proposal.
If the {\em reject} message is for the current proposal, then
the man knows that he must make another proposal. If he is out of proposals, then
he announces that there is no stable marriage with external constraints. Otherwise, he moves on to the next best proposal after
sending out {\em advance} messages to all men with prerequisite proposals.
On receiving an {\em advance} message with woman $w$, the man must ensure that he has made a proposal to woman $w$.
If he has already made a proposal to $w$, then there is nothing to be done; otherwise, 
he skips all proposals till
he gets to his choice which corresponds to $w$. Next, he makes a proposal to $w$ thereby satisfying external constraints.

%

\begin{figure}[htbp]\begin{center}
\begin{tabbing}
x\=xxx\=xxx\=xxx\=xxx\=xxx\= \kill
${\bf P_i}$:: // Process for Man $i$\\
\> {\bf input}\\
\> \>   $mpref$: array[$1$..$n$]  of $1..n;$ // men's preferences\\
\> \>   $mrank$: array[$1$..$n$]  of $1..n;$ // rank of each of the women by man\\
\> \>  // $mrank$ can be derived from $mpref$\\
\> \>   $prerequisite$: array[$1$..$n$] of list of proposals; \\
\> \> // list of proposals that must be executed before $mpref[i]$\\ 
\> {\bf var}\\
\> \>   $G_i:1..n$ initially $1$; // proposal number by $P_i$ \\
\\
\> Upon receiving a message ``initiate'' from environment;\\
\> \>    for each $(m,w) \in prerequisite[G_i]$\\
\> \> \>  send $(``advance" , w)$ to $P_m$;\\
\> \>  send $(``proposal", i)$ to woman $mpref[G_i]$;\\
\\
\> Upon receiving a message $(``reject", j)$:\\
\> \>   {\bf if} $(mpref[G_i] = j)$ {\bf then} // rejected by current partner\\
\> \> \>   {\bf if} $(G_i = n)$ {\bf then}\\
\> \> \> \>  Announce ``no constrained stable marriage possible" ;\\
\> \> \>   {\bf else} \\
\> \> \> \>  $G_i := G_i+1$;\\
\> \> \> \>   for each $(m,w) \in prerequisite[G_i]$\\
\> \> \> \>  \> send $(``advance" , w)$ to $P_m$;\\
\> \> \> \>  send $(``proposal", i)$ to woman $mpref[G_i]$;\\

\\
\> Upon receiving a message $(``advance", q)$:\\
\> \>        {\bf while} $(mrank[q] > G_i)$ \\
\> \> \>        $G_i := G_i+1$\\
\> \> \>         for each $(m,w) \in prerequisite[G_i]$\\
\> \> \>  \>       send $(``advance" , w)$ to $P_m$;\\
\> \>        {\bf endwhile};\\
\> \>  send $(``proposal", i)$ to woman $mpref[G_i]$;\\
\\
${\bf Q_i}$:: // Process for Woman $i$\\
\> {\bf input}\\
\> \>   $wrank$: array[$1$..$n$]  of $1..n;$ // rank of each man by the woman \\
\> {\bf var}\\
\> \>   $partner$: $0..n;$ initially $0$ // current partner \\
\\
\> Upon receiving a message $(``proposal'', j)$:\\
\> \>   {\bf if} $(partner = 0)$ {\bf then}\\
\> \> \>  $partner := j$;\\
\> \>   {\bf else if} $(wrank[j] < wrank[partner]) $ {\bf then}\\
\> \> \> send $(``reject", i)$ to $P_{partner}$;\\
\> \> \>  $partner := j$;\\
\\
\\
{\bf Environment}::\\
 Process that (1) initiates the diffusing computation and \\
 (2) detects Termination\\
 \\
 \> send ``initiate'' message to all $P_i$\\
\> Upon Detecting Termination of Diffusing Computation\\
 \> \> Announce the current assignment as a stable marriage \\
 \> \> satisfying external constraints. Halt\\
\end{tabbing}

\end{center}
\caption{{A diffusing distributed computation algorithm for constrained SMP} for men $P_i$ and women $Q_i$\label{fig:CSMP}}
\end{figure}

Observe that when a man $P_i$ advances, he does not inform his existing partner, if any. Since the number of men and women are same, his partner
will eventually get a proposal from someone who she prefers to $P_i$ if there exists a constrained stable matching.
His partner $q$ can never be matched with $P_j$ such that $q$ prefers $P_i$ over $P_j$. Otherwise, 
we have a blocking pair: both $q$ and $P_i$ prefer each other over their partners.

If there are no external constraints, then there are no {\em advance} messages, and the algorithm is 
a distributed version of the Gale-Shapley algorithm. 
Even in the presence of external constraints, the algorithm shares the following properties with the Gale-Shapley
algorithm. As the algorithms progress, the partner for a man can only get worse and the partner for a woman can only get better.
Both these properties are direct results of the way men send their proposals and the way women respond to proposals.

There are also some crucial differences from the Gale-Shapley algorithm.
In the Gale-Shapley algorithm, once a woman is engaged she continues to be engaged.
For any woman $w$, the predicate that there exists a man such that he is assigned to $w$ is a stable predicate.
As a result, the termination of Gale-Shapley (sequential or distributed version) is easy to detect. When all women have been proposed to, the system
has reached a stable matching. However, due to external constraints, it is not true in CSMP that once a woman is engaged she
continues to stay engaged. The man who she was engaged to, may be required to advance on receiving an {\em advance} message
and then that woman is no longer logically assigned to that man. For the constrained SMP algorithm, we need additional messages to detect termination.
It is the environment process that initiates the computation and detects termination 
of the computation. We assume that a termination detection algorithm such as that of Dijkstra and Scholten \cite{dijkstra1980termination} 
is running in conjunction with the CSMP algorithm. Termination in a diffusing computation corresponds to the condition 
that all processes are passive and there are no messages in-transit. 

We now show that the algorithm in Fig. \ref{fig:CSMP} correctly finds the least assignment (or man-optimal) constrained stable matching whenever it exists.
The correctness follows from the following invariants.
\begin{lemma}
Any assignment $M$ in which $M[i] < G_i$ for any $P_i$ cannot be a constrained stable marriage.
\end{lemma}
\begin{proof}
Initially, the invariant is true because $G_i$ is initialized to $1$ and $M[i]<1$ implies that $P_i$ has not proposed to any one.
There are only two reasons the  $G_i$ variable is incremented.
Either the woman corresponding to the current proposal has sent a {\em reject} or
a man has sent a message to {\em advance} beyond the current woman.
We first consider the case when the current proposal was rejected by the woman $q$.
It is sufficient to show that any assignment in which this man is assigned $q$ cannot be a stable marriage.
Suppose $q$ rejected $P_i$ in favor of $P_j$. If $P_j$ is also assigned to $q$ in $G$, then
it is not a matching. If $P_j$ is assigned to a woman that he proposes to later, then we have that $q$ assigned to
$P_i$ prefers $P_j$ and $P_j$ prefers $q$ to the woman he is assigned.
If $G_i$ is advanced because of an {\em advance} message from $P_j$, then any assignment in which $M[i] < G_i$ does
not satisfy prerequisite constraints due to $\ra_p$.
\end{proof}

To show that the algorithm gives a stable matching on termination, if it exists, we show that
the number of successful proposals is equal to $n$ on termination. 
A proposal is defined to be successful
if it is not rejected by a woman and not advanced over by a man and thereby rejected by the man.
We start the algorithm by each process sending out a proposal. Thus, there are $n$ proposals to start with.
Any proposal that is rejected by a woman leads to another proposal if the reject message is not in transit.
Any proposal that is skipped due to prerequisite constraints also leads to another proposal. 
So either a man
runs out of proposals, or the computation has not terminated until every man has made a successful proposal.
This assertion gives us
\begin{lemma}
If the algorithm announces that the current assignment denotes stable marriage, then the assignment given by $G$
is a stable matching satisfying external constraints, i.e., if $P_i$ is paired with $mpref[i][G_i]$, then 
the assignment satisfies constrained stable matching.
\end{lemma}
\begin{proof}
Since there are no {\em reject} messages, {\em advance} messages, or {\em propose} messages in transit,  we know that
there are $n$ successful proposals. Each successful proposal has the property that
the value of current for $P_i$ equals $j$ iff the value of partner for $Q_j$ equals $i$.
Since any proposal that violates stability is rejected and any proposal that violates external constraints is advanced
we get that the assignment on termination is a stable matching satisfying external constraints.
\end{proof}

We now analyze the message complexity of the algorithm.
Suppose that there are $e$ external constraints, $n$ men, $n$ women
and $m$ unsuccessful proposals.
There are $n$ initiate messages.
For every unsuccessful proposal, the algorithm uses at most one {\em reject} message.
There are exactly $n$ final successful proposals resulting in one message per proposal in the
diffusing computation.
If there are $e$ external constraints (solid edges) across processes), then there are at most $e$ advance messages.
Thus, the messages in the diffusing computation are at most
$n$ {\em initiate messages}, $m$  unsuccessful {\em propose} messages, $m$ {\em reject} messages, $n$ successful {\em propose} messages,
and $e$ {\em advance} messages.
Thus, the total number of messages in the diffusing computation is at most
$2m+2n+e$.

Termination detection algorithms such as Dijkstra and Scholten's requires as many messages as the
application messages in the worst case giving us the overall message complexity of $4m+4n+2e$ messages.
We note here that this message complexity can be reduced by various optimizations such as 
combining the {\em signal/ack} messages of Dijkstra and Scholten's algorithm with application messages.
For example, a {\em reject} message can also serve as an {\em ack} message for a {\em propose} message.
For simplicity, we do not consider these optimizations in the paper.
Since both $m$ and $e$ are $O(n^2)$, we get $O(n^2)$ overall message complexity.
Although the number of unsuccessful proposals can be $O(n^2)$ in the worst case,
they are $O(n \log n)$ on an average for the standard SMP \cite{knuth1997stable} .
Note that each message carries only $O(\log n)$ bits.
\section{Superstable Matching}

In many applications, agents (men and women for the stable marriage problem) may not totally order all their choices. Instead, they
may be indifferent to some choices \cite{IRVING1994261,manlove2002structure}.
We generalize $mpref[i][k]$
 to a set of women instead of a single woman. Therefore, 
 $mrank$ function is not 1-1 anymore. Multiple women may have the same rank.
Similarly, $wrank$ function is not 1-1 anymore. Multiple men may have the same rank.
We now define the notion of blocking pairs for a matching $M$ with ties \cite{IRVING1994261}.
We let $M(m)$ denote the woman matched with the man $m$ and $M(w)$ denote the man
matched with the woman $w$.
In the version, called {\em weakly stable} matching $M$, there is no blocking pair of 
man and woman $(m,w)$ who are not married in $M$ but strictly prefer
each other to their partners in $M$. Formally,
a pair of man and woman $(m,w)$ is {\em blocking for a weakly stable matching} $M$ if they are
not matched in $M$ and\\
\ifdefined\ISBOOK
$(mrank[m][w] < mrank[m][M(m)])  \wedge$
$ (wrank[w][m] < wrank[w][M(w)]. $
\else
$(mrank[m][w] < mrank[m][M(m)])  \wedge$\\
$ (wrank[w][m] < wrank[w][M(w)]. $
\fi

For the weakly stable matching, ties can be broken arbitrarily and any matching that is stable in the resulting instance is also
weakly stable for the original problem.
Therefore, Gale-Shapley algorithm is applicable for the weakly stable matching \cite{IRVING1994261}.
We focus on other forms of stable matching --- superstable and strongly stable matchings.

A matching $M$ of men and women is {\em superstable} if there is no blocking pair $(m,w)$ such that they are not married in $M$
but they either prefer each other to their partners in $M$ or are indifferent with their partners in $M$.
Formally, 
a pair of man and woman $(m,w)$ is {\em blocking for a super stable matching} $M$ if they are
not matched in $M$ and\\
\ifdefined\ISBOOK
$(mrank[m][w] \leq mrank[m][M(m)])  \wedge$
$ (wrank[w][m] \leq wrank[w][M(w)]. $
\else
$(mrank[m][w] \leq mrank[m][M(m)])  \wedge$\\
$ (wrank[w][m] \leq wrank[w][M(w)]. $
\fi


The algorithms for superstable marriage have been proposed in \cite{IRVING1994261,manlove2002structure}.
Our goal is to show that LLP algorithm is applicable to this problem as well.
As before, we will use $G[i]$ to denote the $mrank$ that the man $i$ is currently considering. Initially, $G[i]$ is $1$ for all $i$, i.e.,
each man proposes to all his top choices. We say that $G$ has a superstable matching if there exist $n$ women $w_1, w_2, \ldots w_n$ such that
$\forall i: w_i \in mpref[i][G[i]]]$ and the set $(m_i, w_i)$ is a superstable matching.

We define a bipartite graph $Y(G)$ on the set of men and women with respect to any $G$ as follows.
If a woman does not get any proposal in $G$, then she is unmatched. If she receives multiple proposals then there is an edge
from that woman to all men in the most preferred rank. We say that $Y(G)$ is a perfect matching if 
every man and woman has exactly one adjacent edge in $Y(G)$, 

We claim 
\begin{lemma}
If $Y(G)$ is not a perfect matching, then
there is no superstable matching with $G$ as the proposal vector. 
\end{lemma}
\begin{proof}
If there is a man with no adjacent edge in $Y(G)$  then it is clear that $G$ cannot have a superstable matching.
Now consider the case when a man has at least two adjacent edges. If all the adjacent women for this man have degree one, then
exactly one of them can be matched with this man and other women will remain unmatched.
Therefore, there is at least one woman $w$ who is also adjacent to  another man $m'$. If $w$ is matched with $m$, then $(m',w)$ is a blocking pair.
If $w$ is matched with $m'$, then $(m,w)$ is a blocking pair.
\end{proof}

We now claim that the predicate $B(G) \equiv Y(G) ~ \mbox{is a perfect matching}$ is a lattice-linear predicate.

\begin{lemma}\label{lem:superstable}
If $Y(G)$ is not a perfect matching, then at least one index in $G$ is forbidden.
 \end{lemma}
 \begin{proof}
Consider any man $i$ such that there is no edge adjacent to $i$ in $Y(G)$. This happens when all women that man $i$ has proposed in state $G$ have rejected him.
Consider any $H$ such that
$H[i]$ equals $G[i]$. All the women had rejected man $i$ in $G$. As $H$ is greater than $G$, these women can only have
more choices and will reject man $i$ in $H$ as well.

Now suppose that every man has at least one adjacent edge. 
Let $Z(G)$ be the set of women with degree exactly one. 
If every woman is in $Z(G)$, then we have that $Y(G)$ is a perfect matching because every man has at least one adjacent edge.
If not, consider any man $i$ who is not matched to a woman in $Z(G)$. This means that all the women he is adjacent to have
degrees strictly greater than one. In $H$ all these women would have either better ranked proposals or equally ranked proposals. In either case,
man $i$ would not be matched with any of these women. Hence, $i$ is forbidden.
\end{proof}

We are now ready to present \ref{fig:super-stable}. 
In \ref{fig:super-stable}, we start with the proposal vector $G$ with all components $G[j]$ as $1$. 
Whenever a woman receives multiple proposals, she rejects proposals by men who are ranked lower than anyone who has proposed to her.
We say that a man $j$ is forbidden in $G$, if
every woman $z$ that man $j$ proposes in $G$ is either engaged to or proposed by someone who she prefers to $j$ or is indifferent with respect to $j$.
 \ref{fig:super-stable} is a parallel algorithm because all processes $_j$ such that forbidden($j$) is true can advance in parallel.

\begin{algorithm}
 \SetAlgoRefName{LLP-ManOptimalSuperStableMarriage}
$P_j$: Code for thread $j$\\
 {\bf input}: $mpref[i,k]$: set of int for all $i,k$; $wrank[k][i]$: int for all $k,i$;\\
  {\bf init}:   $G[j] := 1$;\\
{\bf always}:  $Y(j) = mpref[j][G[j]];$\\
 \BlankLine
  {\bf forbidden($j$)}:\\
  \h  $\forall z \in Y(j): \exists i \neq j:  \exists k \leq G[i]: (z \in mpref[i][k]) \wedge (wrank[z][i] \leq wrank[z][j]))$\\
  // all women $z$ in the current proposals from $j$ have been proposed by someone who either they prefer or are indifferent over $j$.\\
 \h  {\bf advance}: $G[j] := G[j]+1; $
\caption{A Parallel Algorithm for Man-Optimal Super Stable Matching \label{fig:super-stable}}
\end{algorithm}

Let us verify that this algorithm indeed generalizes the standard stable marriage algorithm. For the standard stable marriage problem, $mpref[i,k]$ is singleton for all
$i$ and $k$. Hence, $Y(j)$ is also singleton. Using $z$ for the singleton value in $Y(j)$, we get the expression
$ \exists i \neq j:  \exists k \leq G[i]: (z = mpref[i][k]) \wedge (wrank[z][i] < wrank[z][j]))$ which is identical to the stable marriage problem once we
substitute $<$ for $\leq$ for comparing the $wrank$ of man $i$ and man $j$.

When the preference list has a singleton element for each rank as in the classical stable marriage problem, we know that there always exists at least one stable marriage. However, in presence of ties there is no guarantee of existence of a superstable marriage.
Consider the case with two men and women where each one of them does not have any strict preference.
Clearly, for this case there is no superstable marriage.

By symmetry of the problem, one can also get woman-optimal superstable marriage by switching the roles of men and women.
Let $mpref[i].length()$ denote the number of equivalence classes of women for man $i$. When all women are tied for the man $i$, the number of equivalence classes is equal to $1$, and when there
are no ties then it is equal to $n$. 
Consider the distributive lattice $L$ defined as the cross product of $mpref[i]$ for each $i$. 
We now have the following result.
\begin{theorem}
The set of superstable marriages, $L_{superstable}$, is a sublattice of the lattice $L$. 
\end{theorem}
\begin{proof}
From Lemma \ref{lem:superstable}, the set of superstable marriages is closed under meet.
By symmetry of men and women, the set is also closed under join.
\end{proof}

It is already known that the set of superstable marriages forms a distributive lattice \cite{spieker1995set}.
The set of join-irreducible elements of the lattice $L_{superstable}$ forms a partial order
(analogous to the rotation poset \cite{gusfield1989stable}) that can be used to generate all superstable marriages.
Various posets to generate all superstable marriages are discussed in \cite{scott2005study,DBLP:conf/wads/HuG21}

We note that the algorithm \ref{fig:super-stable} can also be used to find the constrained superstable marriage. In particular, the following
predicates are lattice-linear:
\begin{enumerate}
\item
Regret of man $i$ is at most regret of man $j$.
\item
The proposal vector is at least $I$.
\end{enumerate}

\section{Strongly Stable Matching}

A matching $M$ of men and women is {\em strongly stable} if there is no blocking pair $(m,w)$ such that they are not married in $M$
but either (1) both of them prefer each other to their partners in $M$, or (2) one of them prefers the other to his/her partner in $M$ and the other one is
indifferent.
Formally, 
a pair of man and woman $(m,w)$ is {\em blocking for a strongly stable matching} $M$ if they are
not matched in $M$ and\\

\ifdefined\ISBOOK
$((mrank[m][w] \leq mrank[m][M(m)])  \wedge $
$ (wrank[w][m] < wrank[w][M(w)]))$\\
$ \vee ((mrank[m][w] < mrank[m][M(m)])  \wedge $
$(wrank[w][m] \leq wrank[w][M(w)])). $
\else
\h $((mrank[m][w] \leq mrank[m][M(m)])  \wedge $\\
\h $ (wrank[w][m] < wrank[w][M(w)]))$\\ 
$ \vee ((mrank[m][w] < mrank[m][M(m)])  \wedge $\\
\h $(wrank[w][m] \leq wrank[w][M(w)])). $
\fi

As in superstable matching algorithm, we let $mpref[i][k]$ denote the set of women ranked $k$ by man $i$.
As before, we will use $G[i]$ to denote the $mrank$ that the man $i$ is currently considering. Initially, $G[i]$ is $1$ for all $i$, i.e.,
each man proposes to all his top choices. We define a bipartite graph $Y(G)$ on the set of men and women with respect to any $G$ as follows.
If a woman does not get any proposal in $G$, then she is unmatched. If she receives multiple proposals then there is an edge
from that woman to all men in the most preferred rank. For superstable matching, we required $Y(G)$ to be a perfect matching.
For strongly stable matching, we only require $Y(G)$ to contain a perfect matching.

We first note that a strongly stable matching may not exist.
The following example is taken from \cite{IRVING1994261}. \\
\\
$m1: w1, w2$\\
$m2:$ both choices are ties\\
\\
$w1: m2, m1$\\
$w2: m2, m1$

The matching $\{(m1, w1), (m2, w2)\}$ is blocked by the pair $(m2, w1)$: $w1$ strictly prefers $m2$ and $m2$ is indifferent between $w1$ and $w2$.
The only other matching is $\{(m1, w2), (m2, w1)\}$.  This matching is blocked by $(m2, w2)$: $w2$ strictly prefers $m2$ and $m2$ is indifferent between
$w1$ and $w2$.

\remove{
For strongly stable matchings, we will focus only on finding the least proposal vector $G$ such that $Y(G)$ contains a perfect matching.
To that end, we define the predicate $B(G)$ to be true if $G$ is the least vector such that $Y(G)$ has a perfect matching.
We show that whenever there is a perfect matching in $G_1$ and $G_2$ and there is no perfect matching in $G_1 \cap G_2$, then
there exists $G_3 < G_1 \cap G_2$ that has a perfect matching.

}

Consider any bipartite graph with an equal number of men and women. If there is no perfect matching in the graph,
then by Hall's theorem there exists a set of men of size $r$ who collectively are adjacent to fewer than $r$ women.
We define {\em deficiency} of a subset $Z$ of men as $|Z| - N(Z)$ where $N(Z)$ is the {\em neighborhood} of $Z$ (the set of vertices that are adjacent to
at least one vertex in $Z$). The deficiency $\delta(G)$ is the maximum deficiency taken over all subsets of men.
We call a subset of men $Z$ {\em critical} if it is maximally deficient and does not contain any maximally deficient proper subset.
Our algorithm to find a strongly stable matching is simple. 
We start with $G$ as the global state vector with top choices for all men.
If $Y(G)$ has a perfect matching, 
we are done. The perfect matching in $Y(G)$ is a strongly stable matching.
Otherwise, there must be a critical subset of men with maximum deficiency. These set of men must then advance on their proposal number, if possible.
If these men cannot advance, then there does not exist a strongly stable marriage and the algorithm terminates.

\begin{algorithm}
 \SetAlgoRefName{LLP-ManOptimalStronglyStableMarriage}
$P_j$: Code for thread $j$\\
 {\bf input}: $mpref[i,k]$: set of int for all $i,k$; $wrank[k][i]$: int for all $k,i$;\\
  {\bf init}:   $G[j] := 1$;\\
{\bf always}:  $Y(j) = mpref[j][G[j]];$\\
 \BlankLine
  {\bf forbidden($j$)}:\\
  \h $j$ is a member of the critical subset of men in the graph $Y(G)$\\
  \h  {\bf advance}: $G[j] := G[j]+1; $
\caption{A Parallel Algorithm for Man-Optimal Strongly Stable Matching \label{fig:strongly-stable}}
\end{algorithm}

\ref{fig:strongly-stable} is the LLP version of the algorithm proposed by Irving and the interested reader is referred to 
\cite{IRVING1994261}  for the details and the proof of correctness.
Similar to superstable marriages, 
we also get the following result.
\begin{theorem}
The set of strongly stable marriages, $L_{stronglystable}$, is a sublattice of the lattice $L$. 
\end{theorem}
Observe that each element in $L_{stronglystable}$ is not a single marriage but a set of marriages. This is in contrast to 
$L_{superstable}$, where each element corresponds to a single marriage.

\section{Conclusions and Future Work}
We have shown that the Lattice-Linear Parallel Algorithm can solve many problems in the stable marriage literature.
We have shown that the LLP Algorithm can also be converted into an asynchronous distributed algorithm.


In the constrained SMP formulation, we have assumed that $(E, \ra_p)$ is a poset for simplicity.
Our algorithms are applicable when $(E, \ra_p)$ may have cycles. For the general graph $(E, \ra_p)$ we can consider
the graph on strongly connected components which is guaranteed to be acyclic. By viewing each strongly
connected component as a super-proposal in which multiple proposals are made simultaneously, the same
analysis and algorithms can be applied.

We have also derived parallel LLP algorithms for stable matching problems with ties.
Our technique gives an easy derivation of algorithms to find the man-optimal matchings as well as the
sublattice representation of superstable and strongly stable matchings.




%
\bibliography{../../book-opt/fmaster}
\pagebreak
\appendix

\end{document}